\documentclass[10pt,conference]{IEEEtran}
\usepackage{graphicx}
\usepackage{epsfig}
\usepackage{amsmath,amssymb,amsfonts,cite}

\def\Pr{\mathrm{Pr}}

\def\gf{{\sf g}}

\newtheorem{theorem}{Theorem}
\newtheorem{proposition}{Proposition}

\newtheorem{remark}{Remark}

\IEEEoverridecommandlockouts

\begin{document}

\title{\huge The Gaussian Wiretap Channel with a Helping Interferer}


\author{\authorblockN{Xiaojun~Tang\authorrefmark{1},
Ruoheng~Liu\authorrefmark{2},
Predrag~Spasojevi\'{c}\authorrefmark{1}, and
H.~Vincent~Poor\authorrefmark{2} }
\authorblockA{\authorrefmark{1}Rutgers University, \{xtang,spasojev\}@winlab.rutgers.edu}
\authorblockA{\authorrefmark{2}Princeton University, \{rliu,poor\}@princeton.edu}
\thanks{This research was supported by the National Science Foundation
under Grants ANI-03-38807, CNS-06-25637 and CCF-07-28208.}}


\maketitle

\vspace{-0.5cm}
\begin{abstract}
Due to the broadcast nature of the wireless medium, wireless communication is
susceptible to adversarial eavesdropping. This paper describes how
eavesdropping can potentially be defeated by exploiting the superposition nature of the
wireless medium. A Gaussian wire-tap channel with a helping interferer (WTC-HI)
is considered in which a transmitter sends confidential messages to its
intended receiver in the presence of a passive eavesdropper and with the help
of an interferer. The interferer, which does not know the confidential message
assists the confidential message transmission by sending a signal that is
independent of the transmitted message. An achievable secrecy rate and a
Sato-type upper bound on the secrecy capacity are given for the Gaussian
WTC-HI. Through numerical analysis, it is found that the upper bound is close
to the achievable secrecy rate when the interference is weak for symmetric
interference channels, and under more general conditions for asymmetric
Gaussian interference channels.
\end{abstract}

\section{Introduction}\label{sec:intro}

Broadcast and superposition are two fundamental properties of the wireless
medium. Due to the broadcast nature, wireless transmission can be heard by
multiple receivers with possibly different signal strengths. Due to the
superposition nature, a receiver observes a signal that is a superposition of
multiple simultaneous transmissions. From the \textit{secure communication}
point of view, the two properties are interwoven and pose a number of security
issues. In particular, the broadcast nature makes wireless transmission
susceptible to \textit{eavesdropping} since anyone within the communication
range can listen and
possibly extract information. 
A helper can pit one property of the wireless medium against
security issues caused by the other. 
In this paper, we consider the case in which a helper functions as an
interferer to improve the secrecy level of a communication session that is
compromised by a passive eavesdropper. This phenomenon illustrates that
superposition can \emph{enhance} security.

As depicted in Fig. \ref{channel}, we study the problem in which a
transmitter sends confidential messages to an intended receiver with
the help of an interferer, in the presence of a passive
eavesdropper. We call this model the \textit{wiretap channel with a
helping interferer} (WTC-HI). Here, it is desirable to minimize the
leakage of information to the eavesdropper. The secrecy level, i.e.,
the level of ignorance of the eavesdropper with respect to the
confidential message, is measured by the equivocation rate. This
information-theoretic approach was introduced by Wyner
\cite{Wyner:BSTJ:75} for the \textit{wiretap channel} problem, in
which a single source-destination communication is eavesdropped upon
via a degraded channel. Wyner's formulation was generalized by
Csisz{\'{a}}r and K{\"{o}}rner who determined the capacity region of
the broadcast channel with confidential messages
\cite{Csiszar:IT:78}. The Gaussian wiretap channel was considered in
\cite{Leung-Yan-Cheong:IT:78}. The central idea is that the
transmitter uses stochastic encoding \cite{Csiszar:IT:78} to
introduce \emph{randomness}, and hence increase secrecy. In the
WTC-HI model, the helper provides additional randomization via
stochastic encoding without knowing the transmitted message.

\begin{figure}
  \centering
  \includegraphics[width=2.8in]{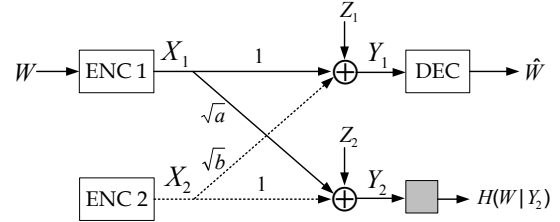}\\
  \caption{A Gaussian wiretap channel with a helping interferer.}\label{channel}
   \vspace{-0.25in}
\end{figure}

In this paper, we give an achievable secrecy rate for the Gaussian WTC-HI under
the requirement that the eavesdropper is kept in total ignorance with respect
to the message for the intended receiver. The results show that the interferer
can indeed increase the secrecy level, and that a positive secrecy rate can be
achieved even when the source-destination channel is worse than the
source-eavesdropper channel. We also describe a power control strategy for
maximizing the achievable secrecy rate.
In addition, we provide a Sato-type upper bound on the secrecy capacity of the
Gaussian WTC-HI. Through numerical analysis, we find that the upper bound is
close to the achievable secrecy rate when the interference is weak for
symmetric interference channels, and under more general conditions for
asymmetric Gaussian WTC-HIs.

Related work includes the multiple access channel with confidential
messages\cite{Liang:IT:06,Liu:ISIT:06,Tekin:IT:07,Tang:ITW:07}, the
interference channel with confidential messages
\cite{Liu:IT:07,Liang:Allerton:07}, and the relay channel with confidential
messages \cite{Oohama:ITW:01,Lai:IT:06,Yusel:CISS:07}. Our achievable scheme
can be considered to be a generalization of the two schemes of
\cite{Tekin:IT:07} and \cite{Lai:IT:06}. The cooperative jamming scheme of
\cite{Tekin:IT:07} considers the situation in which encoder 2 generates independent Gaussian
noise. This scheme does not employ any structure in the transmitted signal. The
noise forwarding scheme of \cite{Lai:IT:06} requires that the interferer's
codewords can always be decoded by the intended receiver which is not necessary
in our scheme. In addition, no work describes a computable upper bound on the
secrecy capacity of the Gaussian WTC-HI.

The remainder of the paper is organized as follows. Section~\ref{sec:model}
describes the system model and problem formulation. Section \ref{sec:results}
states our achievability results. Section \ref{sec:sato} gives the upper bound.
Section~\ref{sec:numerical} illustrate the results through some numerical
examples. The paper is concluded in Section~\ref{sec:conclusions}.

\section{System Model}\label{sec:model}

The system consists of a transmitter, an intended receiver, a
helping interferer, and a passive eavesdropper. The transmitter
sends a confidential message $W$ to the intended receiver with the
help from an \textit{independent} interferer, in the presence of
passive but \textit{intelligent} eavesdropper (who knows both
codebooks). As illustrated in Fig.~\ref{channel}, the channel
outputs at the intended receiver and the eavesdropper can be written
as
\begin{subequations}
\begin{eqnarray}
  Y_{1,k} &=& X_{1,k} +\sqrt{b}X_{2,k} + Z_{1,k}, \label{signalY} \\
  Y_{2,k} &=& \sqrt{a}X_{1,k} + X_{2,k} + Z_{2,k}, \label{signalZ}
\end{eqnarray}
\end{subequations}
for $k=1, \dots, n$, where $\{Z_{1,k}\}$ and $\{Z_{2,k}\}$ are sequences of
independent and identically distributed zero-mean Gaussian noise variables with
unit variances. The channel inputs $X_{1,k}$ and $X_{2,k}$ satisfy average
block power constraints of the form
\begin{equation}\label{power}
  \frac{1}{n}\sum_{k=1}^{n}E[X_{1,k}^2] \leq \bar{P}_1 ~~\mbox{and}~~ \frac{1}{n}\sum_{k=1}^{n}E[X_{2,k}^2] \leq
  \bar{P}_2.
\end{equation}

The transmitter uses encoder 1 to encode confidential message $w \in
\mathcal{W} = \{1,\dots, M\}$ into $x^n$ and sends it to the
intended receiver in $n$ channel uses. A stochastic encoder
$f_1$ is specified by a matrix of conditional
probabilities $f_1(x_{1,k}|w)$, where $x_{1,k} \in \mathcal{X}$, $w
\in \mathcal{W}$, $\sum_{x_{1,k}}f_1(x_{1,k}|w)=1$ for all
$k=1,\dots, n$, and $f_1(x_{1,k}|w)$ is the probability that encoder
1 outputs $x_{1,k}$ when message $w$ is being sent. The helper
generates its output $x_{2,k}$ randomly and can be considered as
using another stochastic encoder $f_2$, which is specified by a
matrix of conditional probabilities $f_{2}(x_{2,k})$ with $x_{2,k}
\in \mathcal{X}_{2}$ and $\sum_{x_{2,k}}f_{2}(x_{2,k})=1.$ Hence,
encoder 1 uses stochastic encoding to introduce \textit{randomness}
and increase secrecy. Additional randomization is provided by the
helper and the secrecy is further increased.

The decoder uses the output sequence $y_1^n$ to compute its estimate
$\hat{w}$ of $w$. The decoding function is specified by a
(deterministic) mapping $\phi: \mathcal{Y}_1^n \rightarrow
\mathcal{W}$.

An ($M,n,P_e$) code for the Gaussian WTC-HI consists of two sets of
$n$ encoding functions ${f_{1,k}}$ and ${f_{2,k}}$, $k=1,\dots,n$
and a decoding function $\phi$ so that its average probability of
error is

\begin{equation}\label{pe}
    P_e=\frac{1}{M}\sum_{w}\Pr\left\{\phi(Y_1^n) \neq w | w
    ~\mbox{sent}\right\}.
\end{equation}

The secrecy level (level of ignorance of the eavesdropper with respect to the
confidential message $w$) is measured by the equivocation rate\footnote{The
secrecy defined by equivocation rate is weak and can be strengthened using
extractor functions without loss of secrecy rate as shown in
\cite{Maurer:EUROCRYPT:00}.} $\frac{1}{n}H(W|Y_2^n).$

A secrecy rate $R_s$ is achievable for the Gaussian WTC-HI if, for
any $\epsilon>0$, there exists an ($M,n,P_e$) code so that

\begin{equation}\label{ach_def1}
    M \geq 2^{nR_s}, ~ P_e \leq \epsilon
\end{equation}
\begin{equation}\label{ach_def2}
\text{and} \qquad  R_s - \frac{1}{n}H(W|Y_2^n) \leq \epsilon \quad
\qquad ~
\end{equation}
for all sufficiently large $n$. The secrecy capacity is the maximal
achievable secrecy rate.

\section{Achievable Secrecy Rate}\label{sec:results}

In this section, we consider an achievable secrecy rate by assuming that the
transmitter and the interferer transmit with powers $P_1 \leq \bar{P}_1$ and
$P_2 \leq \bar{P}_2$, respectively. We address the power control issue in
Subsection \ref{sec:power}.

\subsection{Achievable Secrecy Rate}

\begin{theorem}\label{AchS}
The following secrecy rate is achievable for the Gaussian WTC-HI:
\begin{eqnarray}\label{rate}
R_s(P_1, P_2) = \left\{ \begin{array}{ll}
  0  &\mbox{if $a \geq 1+P_2$,}\\
  R_s^{\mathrm{I}}(P_1, P_2) &\mbox{if $1 \leq a < 1+P_2$,}\\
  R_s^{\mathrm{II}}(P_1, P_2)  &\mbox{if $a<1$,}
  \end{array} \right.
\end{eqnarray}
where $R_s^{\mathrm{I}}(P_1, P_2)$ and $R_s^{\mathrm{II}}(P_1, P_2)$ are given
by
\begin{align}
 &R_s^{\mathrm{I}}(P_1, P_2) = \notag  \\
  &\left\{ \begin{array}{ll}
  \gf(P_1)- \gf(\frac{aP_1}{1+P_2})  &\mbox{if $b \geq 1+P_1$,}\\
  \bigl[\gf(P_1+bP_2)- \gf(aP_1+P_2)\bigr]^{+} &\mbox{if $1 \le  b < 1+P_1$,}\\
  \left[\gf(\frac{P_1}{1+bP_2})- \gf(\frac{aP_1}{1+P_2})\right]^{+}  &\mbox{if $b<1$,}
  \end{array} \nonumber \right.
\end{align}
and
\begin{align}
 &R_s^{\mathrm{II}}(P_1, P_2) = \notag\\
 &\left\{ \begin{array}{ll}
  \gf(P_1)- \gf(\frac{aP_1}{1+P_2})  &\mbox{if $b \geq 1+P_1$,}\\
  \gf(P_1+bP_2)- \gf(aP_1+P_2) &\mbox{if $\beta_1\leq b < 1+P_1$,}\\
    \gf(P_1)- \gf(aP_1)  &\mbox{if $\beta_2\leq b < \beta_1$,}\\
  \gf(\frac{P_1}{1+bP_2})- \gf(\frac{aP_1}{1+P_2})  &\mbox{if $b<\beta_2$,}
  \end{array} \right.\nonumber
\end{align}
with $\gf(x)\triangleq (1/2)\log_2(1+x)$,
\begin{align*}
\beta_1=\frac{1+P_1}{1+aP_1} \quad \text{and} \quad
\beta_2=\frac{a(1+P_1)}{1+aP_1+(1-a)P_2}.
\end{align*}
\end{theorem}

\begin{proof}
We briefly outline the achievability scheme next and omit the
details of the proof.

In the scheme, we use two independent Gaussian codebooks. Encoder 1
uses stochastic codebook $\mathcal{C}_1(2^{nR_1},2^{nR_s}, n)$,
where $2^{nR_1}$ is the size of the codebook, and $2^{nR_s}$ is the number of
confidential messages can be conveyed ($R_s\leq R$). The $2^{nR_1}$ codewords
in codebook $\mathcal{C}_1$ are randomly grouped into $2^{nR_s}$ bins each with
$M=2^{n(R_1-R_s)}$ codewords. In addition, encoder 2 uses codebook
$\mathcal{C}_2(2^{nR_2},1,n)$, where $2^{nR_2}$ is the codebook size and the
whole codebook forms a single bin. To send message $w \in [1,\dots,2^{nR_s}]$,
encoder 1 randomly selects a codeword from the $w$-th bin to send, and encoder
2 randomly selects a codeword from codebook $\mathcal{C}_2$ to send.

The achievable rate given in Theorem \ref{AchS} is derived by using the above
coding scheme and properly choosing the coding parameter triple
$(R_1,R_s,R_2)$.
\end{proof}


\begin{remark}
It is clear that an interference power $P_2$ can benefit secrecy. In
particular, when $P_2$ is sufficiently large, a positive secrecy rate can be
achieved except the case $$a^{-1} \leq b <1.$$

For comparison, we recall that the secrecy capacity of the
Gaussian wiretap channel (when there is no interferer in the
Gaussian WTC-HI model) is
\begin{equation}\label{gwiretap}
     R_s^{\mathrm{WT}}=\left[\gf(P_1)-\gf(aP_1)\right]^{+}
\end{equation}
and positive secrecy rate can be achieved only when $a<1$.

\end{remark}

\subsection{Power Control}\label{sec:power}

Power control is essential to interference management when
accommodating multi-user communications. As for the Gaussian WTC-HI,
power control also plays a critical role. The interferer may need to
control its power so that it does not introduce too much
interference to the primary transmission, while the transmitter may
want to select its power so that the intended receiver is able to
decode and cancel some now helpful interference before decoding the
primary transmission.

In this section, we consider a power control strategy for maximizing
the secrecy rate given in Theorem \ref{AchS}. We consider the cases with $a
\geq 1$ and $a < 1$, respectively. Due to space limitations, we omit
the proof.

\subsubsection{$a \geq 1$}

\begin{proposition}
When $a \geq 1$, the power control scheme for maximizing the secrecy
rate is given by
\begin{eqnarray}\label{powerc1}
(P_1, P_2) = \left\{ \begin{array}{ll}
  (\min\{\bar{P}_1,P_1^{\ast}\}, \bar{P}_2)  &\mbox{if $b > 1, \bar{P}_2 > a-1 $,}\\
  (\bar{P}_1, \min\{\bar{P}_2, P_2^{\ast}\}) &\mbox{if $b < \frac{1}{a}, \bar{P}_2 > \frac{a-1}{1-ab}$, }\nonumber\\
  (0,0)  &\mbox{otherwise,}
  \end{array} \right.
\end{eqnarray}
where $P_1^{\ast}=b-1$ and
\begin{equation}\label{past}
    P_2^{\ast}=\frac{a-1+\sqrt{(a-1)^2+(b^{-1}-a)[a-b+(1-b)a\bar{P}_1]}}{1-ab}.
\end{equation}
\end{proposition}

According to Proposition~1, when $a>1$,  a positive secrecy rate can
be achieved when $b>1$ or $b \leq a^{-1}$ if the interferer's power
$\bar{P}_2$ is large enough. When $b>1$, the interferer uses its
full power $\bar{P}_2$ and the transmitter selects its power to
guarantee that the intended receiver can first decode the
interference (and cancel it). When $b<a^{-1}$, the intended receiver
treats the interference as noise. In this case, the transmitter can
use its full power $\bar{P}_1$ and the interferer controls its power
(below $P_2^{\ast}$) to avoid excessive interference.

\subsubsection{$a < 1$}

\begin{proposition}
When $a < 1$, the power control scheme for maximizing the secrecy
rate is given by
\begin{eqnarray}\label{powerc2}
 \lefteqn{(P_1, P_2) = } \nonumber\\
  &\left\{ \begin{array}{ll}
  (\bar{P}_1, \bar{P}_2)  &\mbox{if $b \geq 1, \bar{P}_1 < b-1$,}\\
  (\bar{P}_1, \bar{P}_2) &\mbox{if $b \geq \frac{1}{a}, \bar{P}_1 \geq b-1, \bar{P}_2 < \frac{1-a}{ab-1}$, }\\
  (P_1^{\ast}, \bar{P}_2) &\mbox{if $b \geq \frac{1}{a}, \bar{P}_1 \geq b-1, \bar{P}_2 \geq \frac{1-a}{ab-1}$, }\\
  (\bar{P}_1,\bar{P}_2) &\mbox{if $1 \leq b < \frac{1}{a}, b-1 \leq \bar{P}_1 < \frac{b-1}{1-ab}$, }\\
  (\bar{P}_1,\min\{\bar{P}_2, P_2^{\ast}\}) &\mbox{if $b < 1, \bar{P}_1 \geq \frac{b-a}{a(1-b)}$, }\\
  (\bar{P}_1,0)  &\mbox{otherwise,}
  \end{array} \right. \nonumber
\end{eqnarray}
where $P_1^{\ast}=b-1$ and $P_2^{\ast}$ is given by (\ref{past}).
\end{proposition}

When $a<1$, positive secrecy rate is always feasible. Here, we
consider the cases when the interferer does not help. First, in the
case when $1 \leq b < a^{-1}$, the transmitter needs to hold its
power if it wants to let the receiver decode some interference.
However, if the transmitter has a large power
($\bar{P}_1>\frac{b-1}{1-ab}$), it would better to use all its power
and request the interferer be silent. In the case when $a< b < 1$,
the receiver treats the interference as noise. If the transmitter
does not have enough power ($\bar{P}_1 <\frac{b-a}{a(1-b)}$), the
interference will hurt the intended receiver more than the
eavesdropper.

\subsection{Power-unconstrained Secrecy Rate}\label{powerun}

A fundamental parameter of wiretap-channel-based wireless secrecy
systems is the secrecy rate when the transmitter has unconstrained
power, which is only related to the channel conditions. For example,
the power-unconstrained secrecy capacity for the Gaussian wiretap
channel is given by
\begin{equation}\label{limit1}
    \lim_{\bar{P}_1 \rightarrow \infty}\left[\gf(\bar{P}_1)-\gf(a\bar{P}_1)\right]^{+}=\frac{1}{2}\left[\log_{2}\frac{1}{a}\right]^{+}.
\end{equation}
After some limiting analysis, we have the following result for the
Gaussian WTC-HI model.

\begin{theorem}

When $a \geq 1$, the achievable power unconstrained secrecy rate for
the Gaussian WTC-HI is
\begin{eqnarray}\label{limit2}
\lim_{\bar{P}_1,\bar{P}_2 \rightarrow \infty}R_s = \left\{
\begin{array}{ll}
  \frac{1}{2}\log_{2}b  &\mbox{if $b > 1$,}\\
  \frac{1}{2}\log_{2}\frac{1}{ab} &\mbox{if $b < \frac{1}{a}$,}\\
  0  &\mbox{otherwise.}
  \end{array} \right.
\end{eqnarray}
When $a<1$, the achievable power unconstrained secrecy rate for
the Gaussian WTC-HI is
\begin{eqnarray}\label{limit3}
\lim_{\bar{P}_1,\bar{P}_2 \rightarrow \infty}R_s = \left\{
\begin{array}{ll}
  \frac{1}{2}\log_{2}b  &\mbox{if $b > \frac{1}{a}$,}\\
  \frac{1}{2}\log_{2}\frac{1}{ab} &\mbox{if $b < 1$,}\\
  \frac{1}{2}\log_{2}\frac{1}{a}  &\mbox{otherwise.}
  \end{array} \right.
\end{eqnarray}
\end{theorem}

When the interference is weak ($b<a^{-1}$ if $a \geq 1$, or $b<1$ if
$a<1$), the interference introduces a gain of $(1/2)\log_2(1/b)$.
When the interference is strong enough ($b>1$ if $a \geq 1$, or
$b>a^{-1}$ if $a<1$), the power-unconstrained secrecy rate is
$(1/2)\log_{2}b$. Note that $(1/2)\log_{2}b$ is the
power-unconstrained secrecy rate if the confidential message is sent
from the interferer to the intended receiver in the presence of the
eavesdropper. This is particularly interesting because we do not
assume that there is a transmitter-interferer channel (which would
enable the interferer to relay the transmission).

\begin{figure}
  \centering
  \includegraphics[width=3.4in]{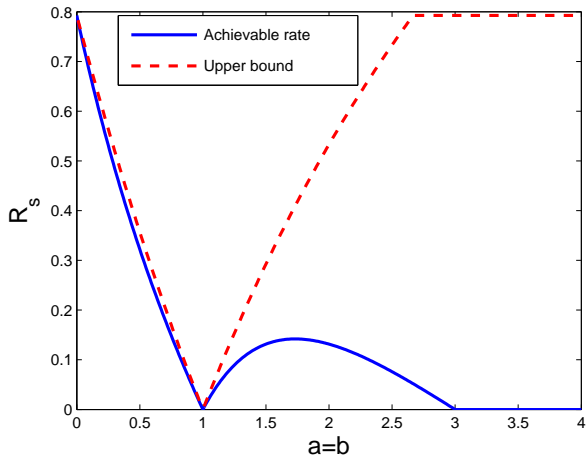}\\
  \caption{Achievable secrecy rate and upper bound versus channel gain $a=b$ for a symmetric channel.}\label{symRs}
  \vspace{-0.1in}
\end{figure}

\section{Sato-Type Upper Bound}\label{sec:sato}

In this section, we first describe a computable Sato-type upper bound for a
general WTC-HI and next evaluate the upper bound for the Gaussian WTC-HI.

It should be noted that the secrecy capacity of the WTC-HI depends
only on the marginal distributions $P_{Y_1|X_1,X_2}$ and
$P_{Y_2|X_1,X_2}$, and not on any further structure of the joint
distribution $P_{Y_1,Y_2|X_1,X_2}$. In fact, the secrecy capacity
is the same for any channel described by
$P_{\tilde{Y}_1,\tilde{Y}_2|X_1,X_2}$ whose marginal distributions
satisfy
\begin{align}
P_{\tilde{Y}_j|X_1,X_2}(y_j|x_1,x_2) &= P_{Y_j|X_1,X_2}(y_j|x_1,x_2)\label{MD-1}
\end{align}
for $j=1,2$ and all $y_1,y_2$ and $x_1,x_2$.

\begin{theorem}\label{Souterbound}
Let $R_u$ denote a Sato-type upper bound
\begin{equation}\label{sato}
    R_u \triangleq \min_{P_{\tilde{Y}_1, \tilde{Y}_2|
    X_1,X_2}}\max_{P_{X_1},P_{X_2}}I(X_1,X_2;\tilde{Y}_1|\tilde{Y}_2).
\end{equation}
Then, the secrecy capacity of WTC-HI satisfies
\begin{equation}\label{eq:up}
    R_s \leq \min \left[R_u, \; \max_{P_{X_1},P_{X_2}}I(X_1;Y_1|X_2)\right].
\end{equation}

\end{theorem}
\begin{proof}
The proof can be found in the Appendix.
\end{proof}

The upper bound assumes that a genie gives the eavesdropper's signal
$\tilde{Y}_2$ to the intended receiver as the side information for decoding
message $W$. Since the eavesdropper's signal $\tilde{Y}_2$ is always a degraded
version of the combined signal $(\tilde{Y}_1,\tilde{Y}_2)$, the wiretap channel
result \cite{Wyner:BSTJ:75} can therefore be used.

Now we consider the evaluate of (\ref{sato}) for the Gaussian WTC-HI. $I(X_1,X_2;\tilde{Y}_1|\tilde{Y}_2)$ is a
function of the transmit power $P_1$, $P_2$ and the noise
covariance $\rho$. Hence, it is denoted as $f(P_1,P_2,\rho)$ and is shown to be
\begin{align} \label{GSato}
&f(P_1,P_2,\rho)=\frac{1}{2} \times \nonumber\\
&\log_2{\frac{(1+P_1+bP_2)(1+aP_1+P_2)-(\rho+\sqrt{a}P_1+\sqrt{b}P_2)^2}{(1-\rho^2)(1+aP_1+P_2)}}.
\end{align}

For any given $\rho$, $f(P_1,P_2,\rho)$ is an increasing function of
both $P_1$ and $P_2$. For any given $P_1$ and $P_2$,
$f(P_1,P_2,\rho)$ is a convex function of $\rho$ and the minimum
occurs when $\rho$ is chosen to be $\rho^{\star}$, which is given by
\begin{equation*}
    \rho^{\star}(P_1,P_2)=\frac{(1+a)P_1+(1+b)P_2+(\sqrt{ab}-1)^2P_1P_2-\sqrt{\Delta}}{2(\sqrt{a}P_1+\sqrt{b}P_2)}
\end{equation*}
where
\begin{align*}
    \Delta&=[(\sqrt{a}-1)^2P_1+(\sqrt{b}-1)^2P_2+(\sqrt{ab}-1)^2P_1P_2]\nonumber\\
          &\times[(\sqrt{a}+1)^2P_1+(\sqrt{b}+1)^2P_2+(\sqrt{ab}-1)^2P_1P_2].
\end{align*}
Therefore, the Sato-type upper bound can be calculated as
\begin{equation*}
    R_u=\min_{\rho}\max_{(P_1,P_2)}f(P_1,P_2,\rho)=f(\bar{P}_1,\bar{P}_2,\rho^*(\bar{P}_1,\bar{P}_2)).
\end{equation*}

\section{Numerical Examples}\label{sec:numerical}

Fig. \ref{symRs} shows the achievable rate and the modified Sato-type upper
bound for a symmetric Gaussian WTC-HI channel ($a=b$). In this example, we
assume that $\bar{P}_1=\bar{P}_2=2$, and $a$ varies from 0 to 4. The achievable
rate $R_s$ first decreases with $a$ when $a<1$; when $1<a \leq 1.73$, $R_s$
increases with $a$ because the intended receiver now can decode and cancel the
interference, while the eavesdropper can only treat the interference as noise;
when $a>1.73$, $R_s$ decreases again with $a$ because the interference does not
affect the eavesdropper much when $a$ is large. The upper bound is good for the
weak interference case when $a \leq 1$. However, when $a>1$ and $a$ is large,
the upper bound is quite loose because too much information is given to the
intended receiver in the genie-aided bound.

In Fig. \ref{Rsb} and Fig. \ref{Rsa}, we present numerical results to show the
achievable rate and the modified Sato-type upper bound for the general
parameter settings of $a$ and $b$, where we again assume that
$\bar{P}_1=\bar{P}_2=2$. In Fig. \ref{Rsb}, we show the secrecy rate versus $b$
when $a$ is fixed to be $0.6$ and $1.2$, respectively. In Fig. \ref{Rsa}, we
show the secrecy rate versus $a$ when $b$ is fixed to be $0.2$ and $1.2$,
respectively. Our numerical results show that the Sato-type upper bound is good
when $ab \leq 1$ (which is consistent with $a \leq 1$ for the symmetric case).
Note that $ab = 1$ corresponds to the degraded case, for which the Sato-type
upper bound is always tight.

\begin{figure}
  \centering
  \includegraphics[width=3.4in]{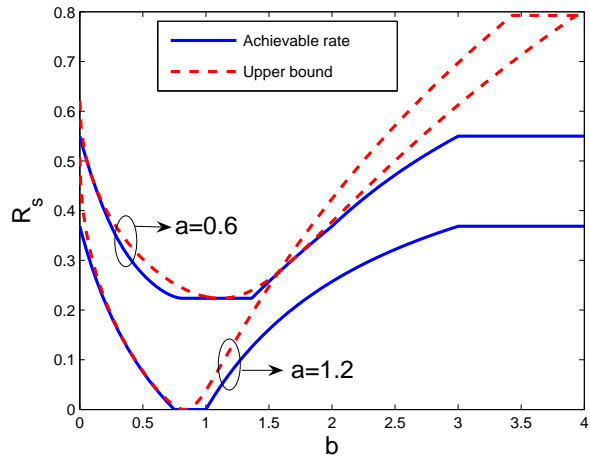}\\
  \caption{Achievable secrecy rate and upper bound versus $b$.}\label{Rsb}
 \vspace{-0.1in}
\end{figure}

\begin{figure}
  \centering
  \includegraphics[width=3.4in]{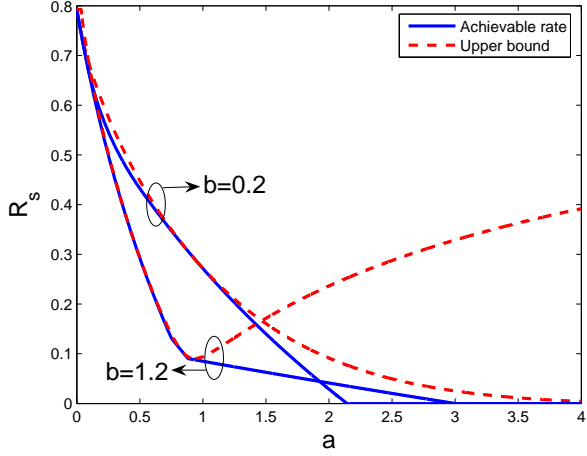}\\
  \caption{Achievable secrecy rate and upper bound versus $a$.}\label{Rsa}
  \vspace{-0.1in}
\end{figure}

\section{Conclusions}\label{sec:conclusions}

In this paper, we have considered the use of the superposition property of the
wireless medium to alleviate the eavesdropping issues caused by the broadcast
nature of the medium. We have studied a Gaussian wiretap channel with a helping
interferer, in which the interferer assists the secret communication by
injecting independent interference. We have given an achievable secrecy rate
and a Sato-type upper bound on the secrecy capacity. The results show that
interference, which seldom offers any advantage for (Gaussian) problems not
involving secrecy, can benefit secret wireless communication.


\appendix
\section{The Sato-type Outer Bound}\label{outerbound}

\begin{proof}[Proof of Theorem \ref{Souterbound}]

The secrecy requirement implies that
\begin{equation}\label{Asecrecy}
    nR_s=H(W) \leq H(W|Y_2^n)+n\epsilon,
\end{equation}
and Fano's inequality implies that
\begin{equation}\label{AFano}
    H(W|Y_1^n) \leq n\epsilon R_1+ h(\epsilon) \triangleq n\delta.
\end{equation}
Based on (\ref{Asecrecy}) and (\ref{AFano}), we have
\begin{align}
    nR_s &\leq H(W|Y_2^n)+n\epsilon \nonumber\\
         &\leq H(W|Y_2^n) - H(W|Y_1^n) + n(\epsilon + \delta) \nonumber \\
         &\leq H(W|Y_2^n) - H(W|Y_1^n, Y_2^n) + n(\epsilon + \delta)\label{A1-1}\\
         &= I(W;Y_1^n|Y_2^n)+ n(\epsilon + \delta)\nonumber\\
         &\leq I(X_1^n, X_2^n;Y_1^n|Y_2^n)+n(\epsilon + \delta) \label{A1-2}\\
         &\leq \sum_{i=1}^{n}I(X_{1,i}, X_{2,i};Y_{1,i}|Y_{2,i})+n(\epsilon + \delta), \label{A1-3}
\end{align}
where (\ref{A1-1}) is due to the fact that conditioning reduces entropy, and
(\ref{A1-2}) follows since $W \rightarrow (X_1^n, X_2^n) \rightarrow (Y_1^n,
Y_2^n)$ forms a Markov chain. Since the secrecy capacity of the WTC-HI depends
only on marginal distributions, we can replace $(Y_1,Y_2)$ with
$(\tilde{Y}_1,\tilde{Y}_2)$ defined by (\ref{MD-1}) and obtain (\ref{sato}).

Now we evaluate (\ref{sato}) for the Gaussian WTC-HI. We let
\begin{align}\label{VVMD}
    \tilde{Y}_1 = X_1 + \sqrt{b} X_2 + \tilde{Z}_1
 \mbox{~and~} \tilde{Y}_2 = \sqrt{a} X_1 + X_2 +
\tilde{Z}_2,
\end{align}
where $\tilde{Z}_1$ and $\tilde{Z}_2$ are arbitrarily correlated
Gaussian random variables with zero-means and unit variances. Let
$\rho$ denote the covariance between $\tilde{Z}_1$ and
$\tilde{Z}_2$, i.e.,
\begin{equation*}
    \mathrm{Cov}(\tilde{Z}_1,\tilde{Z}_2) = \rho.
\end{equation*}
Now, $I(X_1,X_2;\tilde{Y}_1|\tilde{Y}_2)$ can be evaluated as
\begin{align}
    &I(X_1,X_2;\tilde{Y}_1|\tilde{Y}_2) \nonumber\\
         &= I(X_1,X_2;\tilde{Y}_1,\tilde{Y}_2)-I(X_1,X_2;\tilde{Y}_2) \nonumber \\
         &= [H(\tilde{Y}_1,\tilde{Y}_2)-H(\tilde{Y}_1,\tilde{Y}_2|X_1,X_2)]-[h(\tilde{Y}_2)-h(\tilde{Y}_2|X_1,X_2)]\nonumber\\
         &= h(\tilde{Y}_1|\tilde{Y}_2)-h(\tilde{Z}_1|\tilde{Z}_2)\nonumber\\
         &= h(\tilde{Y}_1|\tilde{Y}_2)- \frac{1}{2}\log_2[2\pi e(1-\rho^2)].
\end{align}
By letting
\begin{equation}\label{tc}
    t=\frac{E[\tilde{Y}_1\tilde{Y}_2]}{E[\tilde{Y}_2^2]},
\end{equation}
we have
\begin{align}
    h(\tilde{Y}_1|\tilde{Y}_2) &= h(\tilde{Y}_1- t\tilde{Y}_2|\tilde{Y}_2) \nonumber \\
         &\leq h(\tilde{Y}_1- t\tilde{Y}_2) \label{A2-1}\\
         &\leq \frac{1}{2}\log_2[2\pi e \mathrm{Var}(\tilde{Y}_1-
         t\tilde{Y}_2)]\label{A2-2},
\end{align}
where (\ref{A2-2}) follows from the maximum-entropy theorem and both
equalities in (\ref{A2-1}) and (\ref{A2-2}) hold true when
$(X_1,X_2)$ are Gaussian.

Furthermore, we have
\begin{align*}
\mathrm{Var}(\tilde{Y}_1-t\tilde{Y}_2)=
1+P_1+bP_2-\frac{(\rho+\sqrt{a}P_1+\sqrt{b}P_2)^2}{1+aP_1+P_2}.
\end{align*}
Hence, $I(X_1,X_2;\tilde{Y}_1|\tilde{Y}_2)$ can be evaluated by
(\ref{GSato}).
\end{proof}

\bibliographystyle{IEEEtran}
\bibliography{MacFC}

\end{document}